\documentclass[%
 reprint,
 amsmath,amssymb,
 aps,
pra,
floatfix
]{revtex4-2}

\usepackage{graphicx}
\usepackage{multirow}
\usepackage{amsthm}

\newcommand{\mat}[1]{\mathbf{#1}}
\newcommand\eps{\varepsilon}
\newcommand\ee{\mathrm{e}}

\newcommand\ii{\mathrm{i}}

\newcommand\nm{\mathrm{~nm}}

\newcommand\TT{\mathrm{T}}
\newcommand\diag{\mathop{\rm diag}\limits}

\newtheorem{lemma}{Lemma}

\newtheorem*{theorem*}{Theorem}

\begin{document}

\title{Coherent perfect absorption and lasing \\in bimodal Fabry--P\'erot interferometers}

\author{Dmitry A. Bykov}
 \email{bykovd@gmail.com}
\author{Evgeni A. Bezus}
\author{Leonid L. Doskolovich}

\affiliation{Samara National Research University, 34 Moskovskoye shosse, Samara 443086, Russia}
\affiliation{Image Processing Systems Institute, National Research Centre ``Kurchatov Institute'', 151 Molodogvardeyskaya st., Samara 443001, Russia}

\date{\today}
			
\begin{abstract} 
Bimodal Fabry--P\'erot interferometer is a model generalizing the conventional Fabry--P\'erot interferometer, in which not one but two kinds of waves propagate between the interfaces. 
Here, we study coherent perfect absorption (CPA) and lasing at threshold in bimodal Fabry--P\'erot interferometers.
We show that CPA and lasing appear only in certain ``allowed'' regions in the parameter space, which, as we demonstrate, are described by closed-form inequalities imposed on the elements of the scattering matrix of the interferometer interfaces.
We demonstrate topologically governed annihilation of CPA points as they approach the boundary of a CPA-allowed region.
In the particular case when the absorption losses tend to zero,
the presented model describes the formation of bound states in the continuum exactly at the CPA annihilation points.
The presented analytical model is in perfect agreement with the rigorous numerical simulation results of high-contrast gratings and ridge resonators implementing the bimodal Fabry--P\'erot interferometer model.
\end{abstract}

\maketitle

\section{Introduction}

In the last fifteen years, coherent perfect absorption (CPA) of light attracted a lot of research attention~\cite{Chong:prl:2010, Baranov:2017:nrm, Krasnok:2019:aop, Wang:2021:s}.
CPA can be considered as a generalization of the perfect absorption effect~\cite{Radi:2015:prappl, Kats:2016:lpr} to the case when the structure is illuminated by several coherent waves, the amplitudes and phases of which are chosen so that no scattered radiation appears.
It is worth noting that an absorber exhibiting CPA can be regarded as a time-reversed laser operating exactly at threshold~\cite{Chong:prl:2010, Baranov:2017:nrm, Krasnok:2019:aop}.
Structures providing perfect absorption and coherent perfect absorption are promising for applications in 
sensing~\cite{Liu:2010:nl, Vasic:2014:jap},
all-optical image processing~\cite{Papaioannou:2016:lsa, Papaioannou:2017:acsp}, and
linear-optics implementation of switches, modulators, and logical gates~\cite{Fang:2014:apl, Fang:2015:lsa, Papaioannou:2016:lsa, Papaioannou:2016:aplp}.

Coherent perfect absorption was investigated in various photonic structures including 
thin films and multilayers~\cite{Chong:prl:2010, Villinger:2015:ol, Pye:2017:ol},
diffraction gratings~\cite{DuttaGupta:2012:ol, Yoon:2012:prl},
metasurfaces~\cite{Zhang:2012:lsa, Fang:2014:apl, Fang:2015:lsa, Zhu:2016:apl}, 
graphene-based structures~\cite{Fan:2014:ol, Zanotto:2016:aplp}, and disordered media~\cite{Pichler:2019:n}.
CPA was also studied for several integrated platforms including silicon photonics and plasmonics~\cite{Bruck:2013:oe, Ignatov:2016:adp, Bezus:2022:ol}.
In $\mathcal{P}\mathcal{T}$-symmetric structures, CPA and lasing can occur simultaneously~\cite{Longhi:2010:pra, Chong:2011:prl}.
For specially engineered structures, CPA can occur at an exceptional point~\cite{Wang:2021:s, Horner:2024:prl}.

For most of the structures investigated in the above referenced works, coherent perfect absorption can be described using the Fabry--P\'erot interferometer model.
In the simplest case, a Fabry--P\'erot interferometer comprises a slab having 
the interfaces coated with semi-transparent mirrors, which can be dielectric or metallic [Fig.~\ref{fig:1}(a)].
When the interferometer is illuminated with a plane wave,
the field inside the slab is the superposition of one plane wave propagating upwards and one plane wave propagating downwards.

\begin{figure}
	\centering
		\includegraphics{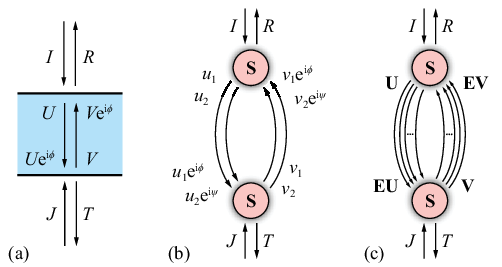}
	\caption{\label{fig:1}Conventional (a), bimodal (b), and generalized (c) Fabry--P{\'e}rot interferometers.	}
\end{figure}

In a \emph{generalized} Fabry--P\'erot interferometer [Fig.~\ref{fig:1}(c)], more than one kind of modes (or waves) propagate between the ``interfaces'' that couple these modes with each other and with the incident and scattered waves above and below the structure~\cite{Karagodsky:2010:oe, Karagodsky:2011:ol, Orta:2016:1, Orta:2016:2, Tibaldi:2018}. 
The word ``interfaces'' is enquoted here, since this part of the structure may be no longer flat containing, e.g., diffraction gratings or another complex surface relief.
Such a model allows one to describe resonant optical properties of a wide range of photonic structures including 
guided-mode resonant gratings~\cite{Bykov:2019:pra},
high-contrast gratings~\cite{Karagodsky:2010:oe, Karagodsky:2011:ol, Orta:2016:1, Orta:2016:2, Tibaldi:2018, Ovcharenko:2020:prb, Bykov:2024:pra},
and integrated ridge resonators~\cite{Bezus:2018:pr, Bykov:2020:n}.

Of particular interest are \emph{bimodal} Fabry--P\'erot interferometers, in which two kinds of modes propagate between the interfaces [Fig.~\ref{fig:1}(b)].
Such a model, permitting a detailed theoretical analysis of the resonant optical properties, was, in particular, used to describe bound states in the continuum (BICs) in various photonic structures~\cite{Bezus:2018:pr, Bykov:2019:pra, Bykov:2020:n, Ovcharenko:2020:prb, Bezus:2021:n, Bykov:2024:pra}.
Notably, it is exactly the bimodal case, for which simple closed-form expressions for BIC positions can be obtained~\cite{Bezus:2018:pr, Bykov:2019:pra, Bykov:2020:n}.
In this paper, we present a theory for CPA and lasing in bimodal Fabry--P\'erot interferometers.
We demonstrate that these effects appear only in certain parts of the parameter space and derive explicit inequalities defining these parts.

\section{Generalized Fabry--P\'erot interferometer model}

To obtain the equations describing the generalized Fabry--P\'erot interferometer, let us denote the complex amplitudes of the incident, reflected, and transmitted waves as $I$, $R$, and $T$.
We will also require another plane wave having the amplitude $J$, which is incident on the structure from below [see Fig.~\ref{fig:1}(c)].
We will assume that there are $n$ kinds of modes (waves) with the wavenumbers $k_{z,j},j=1,\ldots,n$ propagating from one interface to another.
The complex amplitudes of the modes propagating downwards, $u_j$, will be arranged to a column vector $\mat{U}$ ,
defining the amplitudes of the modes at the upper interface.
When they propagate to the lower interface, their phases will change resulting in the complex amplitude vector $\mat{E}\mat{U}$, where $\mat{E}$ is an $n\times n$ diagonal matrix:
$$ 
\mat{E} = \diag_j \ee^{\ii \phi_j}.
$$
Denoting the distance between the interfaces by $w$, 
we can express the phases as $\phi_j = k_{z,j} w$.
These phases will be real numbers if we assume that there is no loss in the medium separating the interfaces.
Similarly, the waves propagating upwards have the amplitudes $\mat{V} = [v_1, v_2, \ldots, v_n]^{\rm T}$ and $\mat{E}\mat{V}$ at the lower and upper interfaces, respectively.

To describe coupling of the introduced waves and modes, we will use the $(n+1)\times(n+1)$ scattering matrix of the upper interface:
\begin{equation}
\label{eq:S}
\mat{S} =
	\begin{bmatrix}
		\mat{r} & \mat{t}\\
		\mat{t}^\TT & r_0\\
  \end{bmatrix}.
\end{equation}
Here, the $n\times n$ submatrix $\mat{r}$ describes how the considered $n$ modes reflect from the interface,
vector $\mat{t}$ describes how these modes leak out of the structure contributing to the reflected field,
and $r_0$ is the coefficient describing how the incident wave $I$ is reflected from the interface.
In this paper, we will be considering structures with a horizontal symmetry plane,
thus the same scattering matrix describes the lower interface as well.
It is important to note that the matrix $\mat{S}$ is not unitary since there must be loss in a structure to demonstrate CPA, and gain is required for lasing.
We also note that the scattering matrix is assumed to be symmetric ($\mat{S}^\TT=\mat{S}$), which is the case for reciprocal structures in integrated optics and for reciprocal gratings with a vertical symmetry plane~\cite{Gippius:2005:prb}.

Using the above-introduced notation, we can couple the waves at the upper and lower interfaces,
which gives us a system of two matrix equations:
\begin{equation}
	\label{eq:main}
	\begin{bmatrix}
		\mat{U} \\ R
	\end{bmatrix}
	=
	\mat{S}
	\begin{bmatrix}
		\mat{E} \mat{V} \\I
	\end{bmatrix},
	\;\;\;\;\;\;
	\begin{bmatrix}
		\mat{V} \\ T
	\end{bmatrix}
	=
	\mat{S}
	\begin{bmatrix}
		\mat{E} \mat{U} \\ J
	\end{bmatrix}.
\end{equation}
This system describes light scattering by the considered structure and allows one to find its reflected and transmitted spectra.
In particular, it describes various resonant effects such as the formation of high-Q resonances~\cite{Bezus:2018:pr, Bykov:2019:pra, Bykov:2020:n, Ovcharenko:2020:prb, Bezus:2021:n, Bykov:2024:pra}.
We will use these equations to describe coherent perfect absorption and lasing at threshold.
Interestingly, the equations describing lasing have simpler form, thus we discuss it first before moving to the CPA.

\section{Lasing condition}

\subsection{General case}

Lasing states appear when the amplitudes of the scattered waves $R$ and $T$ are non-zero in the absence of the incident light, i.\,e., at $I$ = $J$ = 0.
At the same time, the light frequency $\omega$ is required to be real, which is possible only if there is gain in the structure and, thus, the scattering matrix $\mat{S}$ is non-unitary.
As above, we assume that the media between the upper and lower interfaces is lossless, thus, the gain media has to be incorporated inside the ``interfaces'', which makes the considered structure geometry similar to the external cavity lasers.
It is important to note that in this section, we will assume that the scattering matrix $\mat{S}$ of the interfaces captures all the ``laser physics'' of the structure.

Since the considered Fabry--P\'erot interferometer has a horizontal symmetry plane, all the solutions of the Maxwell's equations, including those describing lasing and CPA, are either symmetric or antisymmetric with respect to this plane.
We will study these two cases separately.

Symmetric lasing occurs when $R=T$ and $I=J=0$ in Eq.~\eqref{eq:main}, 
thus, light is emitted to the substrate and superstrate evenly and in phase.
Besides, due to symmetry, the amplitudes of the modes propagating upwards and downwards also coincide: $\mat{U}=\mat{V}$.
This allows us to rewrite the two equalities of Eq.~\eqref{eq:main} as a single one:
\begin{equation}
\label{eq:lasing:sym:1}
\mat{U} = \mat{r}\mat{E} \mat{U}.
\end{equation}
Here, we used the notation introduced in Eq.~\eqref{eq:S} for the sub-blocks of $\mat{S}$.
An expression for the emitted light amplitude also follows from Eq.~\eqref{eq:S}: $R = T = \mat{t}^\TT \mat{U}$, but it is of little use for the further analysis.
Equation~\eqref{eq:lasing:sym:1}, on the other hand, is quite important;
we can think of it as of the requirement for the matrix $\mat{r}\mat{E}$ to have a unit eigenvalue with the corresponding eigenvector $\mat{U}$.
We can rewrite this condition in terms of the matrix determinant:
\begin{equation}
\label{eq:lasing:sym:2}
\det(\mat{r}\mat{E}-\mat{I}) = 0,
\end{equation}
where $\mat{I}$ is the identity matrix.
This condition is the symmetric lasing condition for generalized Fabry--P\'erot interferometers.

The antisymmetric lasing condition is obtained similarly:
we simply rewrite Eq.~\eqref{eq:main} with $R=-T$, $\mat{V}=-\mat{U}$, and $I=J=0$ and obtain
\begin{equation}
\label{eq:lasing:asym:2}
\det(\mat{r}\mat{E}+\mat{I}) = 0.
\end{equation}

\subsection{Bimodal case}
In this paper, we are paying particular attention to the case of bimodal Fabry--P\'erot interferometers, in which two ($n=2$) kinds of modes propagate between the interfaces. 
This case is worth studying since it admits a more detailed analysis, as we demonstrate below.

Let us denote the elements of the matrices $\mat{S}$ and $\mat{E}$ as follows:
\begin{equation}
\label{eq:S3}
\mat{S} = 
	\begin{bmatrix}
		r_{11} & r_{12} & t_1\\
		r_{12} & r_{22} & t_2\\
		t_1 & t_2 & r_0
  \end{bmatrix}; \;\;\;\;\;
\mat{E} = 
	\begin{bmatrix}
		\ee^{\ii \phi}& 0 \\
		0 & \ee^{\ii \psi}
  \end{bmatrix}.
\end{equation}
With this notation, we can rewrite the symmetric lasing condition of Eq.~\eqref{eq:lasing:sym:2} as
\begin{equation}
\label{eq:lasing:n2}
\det
	\begin{bmatrix}
		r_{11} \ee^{\ii \phi} - 1 & r_{12}\ee^{\ii \psi}    \\
		r_{12} \ee^{\ii \phi}     & r_{22}\ee^{\ii \psi} - 1 \\
  \end{bmatrix} = 0.
\end{equation}
Now we expand the determinant in the last equation and rewrite it as 
\begin{equation}
\label{eq:lasing:n3}
(r_{11} r_{22}-r_{12}^2) \ee^{\ii\phi} \ee^{\ii\psi}
-r_{11} \ee^{\ii \phi}
-r_{22} \ee^{\ii \psi}
+ 1 = 0.
\end{equation}

Let us show that this lasing condition can be reformulated as an explicit inequality relating only the elements of the scattering matrix $\mat{S}$.
In order for lasing to be possible, the last equation must be solvable for some real values of $\phi$ and $\psi$.
To rewrite this requirement, we will use the following theorem, which is proved in Appendix~\ref{app:A}:
\begin{theorem*}
Equation
$$
a \ee^{\ii \phi}+b \ee^{\ii \psi}+c \ee^{\ii \phi}\ee^{\ii \psi}+d = 0
$$
with $a,b,c$, and $d$ being complex numbers is solvable in real $\phi$ and $\psi$ if and only if
$$
\left||a|^2+|b|^2-|c|^2-|d|^2\right| \leq 2|ab-cd|.
$$
\end{theorem*}
Using this theorem, we rewrite the fact that lasing condition~\eqref{eq:lasing:n3} is solvable
 in the following equivalent form:
\begin{equation}
\label{eq:main:lasing}
\left||r_{11}|^2 + |r_{22}|^2 - |r_{11} r_{22}-r_{12}^2|^2 - 1\right| \leq 2|r_{12}|^2.
\end{equation}
This inequality is a symmetric lasing condition for bimodal Fabry--P\'erot interferometers
since we obtained it from the symmetric lasing condition~\eqref{eq:lasing:sym:2}.
However, one can easily show that the antisymmetric condition obtained from Eq.~\eqref{eq:lasing:asym:2} has exactly the same form.
Inequality~\eqref{eq:main:lasing} is one of two main results of this paper.

We have shown that lasing in bimodal Fabry--P\'erot interferometers is possible only if the condition of Eq.~\eqref{eq:main:lasing} holds, which makes it a \emph{necessary} condition for lasing, i.\,e., it defines a part of the parameter space where lasing is possible.
To find a particular point where lasing occurs, we have to tune two parameters to satisfy Eq.~\eqref{eq:lasing:n3}.
The light frequency and thickness of the structure might be such parameters.
Indeed, according to the theorem, some real values of $\phi$ and $\psi$ satisfying Eq.~\eqref{eq:lasing:n3} exist.
Then, we can find the corresponding values of $w$ and $\omega$ by solving the system of two equations: 
$\phi = w k_{z,1}(\omega)$; $\psi = w k_{z,2}(\omega)$.
This is possible since $\mat{S}$ is not the scattering matrix of the whole structure but that of its interfaces; therefore, its elements depend on frequency quite slow and $\phi$ and $\psi$ can be assumed to be frequency-independent.

\section{Coherent perfect absorption condition}\label{sec:cpa}

\subsection{General case}
Coherent perfect absorption occurs when for some incident wave amplitudes $I$ and $J$ there appears no scattered light: $R=T=0$.
Similarly to lasing, CPA solutions in symmetric structures have either symmetric or antisymmetric field distributions.

Symmetric CPA solutions are described by Eq.~\eqref{eq:main} with $I=J=1$, $\mat{U}=\mat{V}$, and $R=T=0$:
$$
\begin{aligned}
\mat{U} &= \mat{r}\mat{E} \mat{U} + \mat{t},\\
0 &= \mat{t}^\TT \mat{E}\mat{U} + r_0.
\end{aligned}
$$
Solving the first equation for the vector $\mat{U}$ and substituting it into the second one yields the symmetric CPA condition:
$$
\mat{t}^\TT \mat{E}(\mat{I} - \mat{r} \mat{E})^{-1} \mat{t} + r_0 = 0.
$$
Using the matrix determinant lemma, we can rewrite this condition in terms of the matrix determinant:
\begin{equation}
\label{eq:cpa:sym}
\det\left((\mat{r}-r_0^{-1} \mat{t} \mat{t}^\TT)\mat{E} - \mat{I} \right) = 0.
\end{equation}

The antisymmetric CPA condition is obtained similarly:
\begin{equation}
\label{eq:cpa:asym}
\det\left((\mat{r}-r_0^{-1} \mat{t} \mat{t}^\TT)\mat{E} + \mat{I} \right) = 0.
\end{equation}

\subsection{Time reversal reasoning}
Lasing states are related with CPA states by time reversal~\cite{Chong:prl:2010}.
In this subsection, we show how this fact applies to the presented theory.

Assume we have a monochromatic solution to the Maxwell's equations, $\vec{E}(x,y,z)$ and $\vec{H}(x,y,z)$, at a real frequency $\omega$ for some structure having the permittivity function $\eps(x,y,z)$.
The complex-conjugated fields, $\vec{E}^*(x,y,z)$ and $-\vec{H}^*(x,y,z)$, also satisfy the Maxwell's equations yet for the structure with the conjugated permittivity $\eps^*(x,y,z)$~\cite{Haus:1984}.
Speaking about plane waves, conjugation not only conjugates their complex amplitudes but also reverses the direction of each wave, thus incident waves become scattered ones and vice versa.
Speaking about scattering matrices, time reversal means that by knowing the scattering matrix $\mat{S}$ for a structure given by $\eps(x,y,z)$ we can find the scattering matrix for the ``conjugated'' structure $\eps^*(x,y,z)$.
Indeed, assume that the amplitudes of the incident waves $\mat{a}$ are related with the scattered ones $\mat{b}$ as $\mat{b} = \mat{S} \mat{a}$.
We rewrite this equation as $\mat{a}^* = (\mat{S}^*)^{-1} \mat{b}^*$
and now the matrix $(\mat{S}^*)^{-1}$ is exactly the scattering matrix for the ``conjugated'' structure where the incident waves are the conjugation of the scattered waves from the original structure~\cite{Haus:1984}.

Importantly, if the structure is described by several coupled parts, each having its own scattering matrix, the time reversal transforms each scattering matrix according to the presented law. In the considered case,
time reversal replaces $\mat{S}$ with $(\mat{S}^*)^{-1}$, whereas $\mat{E}$, which could be arranged to a scattering matrix as well, stays intact since $(\mat{E}^*)^{-1}=\mat{E}$.

Assume we found the matrices $\mat{S}$ and $\mat{E}$ satisfying the lasing condition~\eqref{eq:lasing:sym:2}.
If these matrices could be implemented by some structure $\eps(x,y,z)$, 
the ``conjugated'' structure $\eps^*(x,y,z)$ would exhibit CPA behavior since the scattered waves of the lasing state would become the incident waves and the absence of incident waves in the lasing states, after conjugation, would provide the lack of the scattered radiation.
The conjugated structure will be described by the generalized Fabry--P{\'e}rot interferometer model with the matrices $(\mat{S}^*)^{-1}$ and $\mat{E}$, which should satisfy the CPA condition~\eqref{eq:cpa:sym}.
Let us verify this fact and obtain Eq.~\eqref{eq:cpa:sym} from Eq.~\eqref{eq:lasing:sym:2}.

We start with Eq.~\eqref{eq:lasing:sym:2} with $\mat{S}$ replaced by $(\mat{S}^*)^{-1}$.
Therefore, $\mat{r}$ should be replaced with the upper left $n\times n$ block of $(\mat{S}^*)^{-1}$, 
which, with the use of the blockwise matrix inversion formula, can be found to be
$[(\mat{r} - r_0^{-1} \mat{t} \mat{t}^\TT)^*]^{-1}$.
After substituting it as $\mat{r}$ into Eq.~\eqref{eq:lasing:sym:2}, simple transformations indeed lead to Eq.~\eqref{eq:cpa:sym}.

Therefore, the CPA condition is exactly the lasing condition, in which the scattering matrix $\mat{S}$ is replaced with its conjugated inverse.

\begin{figure*}
	\centering
		\includegraphics{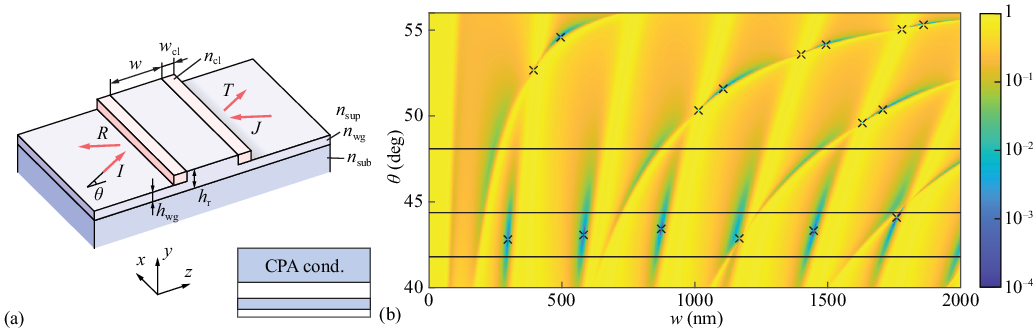}
	\caption{\label{fig:2}(a)~Ridge resonator with absorbing claddings.	
	(b)~Symmetric reflectance $|R_{\rm sym}|^2$ of the ridge resonator. 
	Black crosses show the CPA points.
	The inset shows the part of the parameter space, in which the CPA condition~\eqref{eq:main:cpa} is fulfilled.
	}
\end{figure*}

\subsection{Bimodal case}
A straightforward way to obtain the CPA condition for a bimodal Fabry--P\'erot interferometer suggests
using Eqs.~\eqref{eq:S} and~\eqref{eq:S3} in Eq.~\eqref{eq:cpa:sym} and applying the theorem from Appendix~\ref{app:A}. 
Let us, however, follow a shorter path based on time reversal.

As we noted in the previous subsection, the CPA condition can be obtained from the lasing condition by formally replacing the elements of the matrix $\mat{S}$ with the corresponding elements of $(\mat{S}^*)^{-1}$.
The scattering matrix inverse is easily written in terms of the adjugate matrix:
\begin{equation}
\label{eq:sinv}
\mat{S}^{-1} = 
	\frac1{\det\mat S}
	\begin{bmatrix}
		r_0 r_{22} - t_2^2      & t_1 t_2 -r_0 r_{12}     & r_{12} t_2 - r_{22} t_1 \\			
		t_1 t_2 -r_0 r_{12}     & r_0 r_{11}-t_1^2        & r_{12} t_1 - r_{11} t_2 \\
		r_{12} t_2 - r_{22} t_1 & r_{12} t_1 - r_{11} t_2 & r_{11} r_{22}-r_{12}^2
  \end{bmatrix}.
\end{equation}
Performing the $\mat{S}\to(\mat{S}^{-1})^*$ replacement in Eq.~\eqref{eq:main:lasing} yields, after simplification, the second main result of this paper, namely, the CPA condition for bimodal Fabry--P\'erot interferometers:
\begin{equation}
\label{eq:main:cpa}
\begin{aligned}
&\left||r_0 r_{22} - t_2^2|^2 + |r_0 r_{11}-t_1^2|^2 - |r_0|^2 - |\det\mat{S}|^2\right| 
\\ &\hspace{14em}\leq 2|t_1 t_2 -r_0 r_{12}|^2.
\end{aligned}
\end{equation}

Similarly to Eq.~\eqref{eq:main:lasing}, this condition describes both symmetric and antisymmetric CPA
and defines a part of the parameter space where CPA might be encountered.
As in the case of lasing, finding a particular CPA point in the specified part of the parameter space requires tuning two parameters.
Despite a rather complex form of Eq.~\eqref{eq:main:cpa}, one can easily check whether it is satisfied for a particular structure.
Indeed, the inequality is directly imposed on the scattering matrix elements, which can be readily calculated using conventional numerical simulation techniques, as we demonstrate in the following section.

\section{Simulation results}

In this section, we will numerically investigate several photonic structures implementing the bimodal Fabry--P{\'e}rot interferometer.
To demonstrate the generality of the obtained theoretical results, we consider two distinct examples:
an integrated ridge resonator on the dielectric slab waveguide platform in Subsection~\ref{ssec:ridge}
and a suspended high-contrast grating with one-dimensional periodicity in Subsections~\ref{ssec:hcg}--\ref{ssec:hcg:lowabs}.

For both structures, we will focus on the symmetric CPA, however, all the effects demonstrated below appear in the antisymmetric case as well.
To investigate symmetric CPA, we will be calculating the symmetric reflection coefficient $R_{\rm sym}$.
This coefficient is the complex amplitude of the reflected wave for the case when \emph{two} unit-amplitude incident waves ($I=J=1$)
impinge on the structure from both sides (see Fig.~\ref{fig:1}).
To obtain symmetric CPA, these waves are assumed coherent and having the same phases.
One can easily show that the symmetric reflection coefficient $R_{\rm sym}$ is related with the ``conventional'' reflection and transmission coefficients (obtained at $I=1$ and $J=0$) as follows: $R_{\rm sym} = R+T$.
All simulations are performed using the Fourier modal method (FMM)~\cite{Moharam:1995:josaa, Li:1996:josaa, Li:1996:josaa2}, an established numerical tool for solving the Maxwell's equations for periodic structures. 
For simulating integrated-optics structures in Subsection~\ref{ssec:ridge}, an aperiodic formulation of the FMM was used~\cite{Silberstein:2001:josaa, Hugonin:2005:josaa}.

\subsection{CPA in integrated ridge resonator}\label{ssec:ridge}

As the first example, we will consider an integrated ridge resonator~\cite{Zou:2015:lpr, Bezus:2018:pr, Nguyen:2019:lpr} with absorbing claddings shown in Fig.~\ref{fig:2}(a) and operating at the free-space wavelength $\lambda = 630\nm$.
The resonator is located on a single-mode dielectric slab waveguide with thickness $h_{\rm wg} = 80\nm$ and refractive index $n_{\rm wg} = 3.32$.
The refractive indices of the superstrate and substrate amount to $n_{\rm sup} = 1$ and $n_{\rm sub} = 1.45$.
The thickness of the waveguide in the ridge region is $h_{\rm r} = 110\nm$.
The absorbing claddings of the resonator are constituted by metal strips placed on both sides of the ridge and having the thickness $h_{\rm r} - h_{\rm wg} = 30\nm$ and width $w_{\rm cl} = 85\nm$.
The dielectric permittivity of the strips equals $n_{\rm cl}^2 = -18.12 + 0.51\ii$, which corresponds to silver at the considered wavelength.

We will consider the case of symmetric oblique incidence of two TE-polarized guided modes on the structure as shown in Fig.~\ref{fig:2}(a).
At the presented parameters, in the considered angle of incidence range $\theta \in [40^\circ, 56^\circ]$, the scattered field also contains only TE-polarized guided modes, whereas in the ridge region, the slab waveguide has a greater thickness and thus supports both TE- and TM-polarized modes~\cite{Bezus:2018:pr}, which are coupled at the ridge interfaces where the metal strips lie.
Therefore, the considered integrated structure indeed implements the bimodal Fabry--P{\'e}rot interferometer model.

Figure~\ref{fig:2}(b) shows the symmetric reflectance of the ridge resonator $|R_{\rm sym}|^2$ vs. the angle of incidence $\theta$ and ridge width $w$.
For the considered example, the CPA condition~\eqref{eq:main:cpa} is fulfilled only in a part of the considered parameter space.
Since the CPA condition does not depend on $w$, the CPA can occur in two horizontal ``stripes'' [see the inset to Fig.~\ref{fig:2}(a)].
The CPA points, which were calculated by numerically solving the equation $R_{\rm sym} = 0$, are shown with black crosses in Fig.~\ref{fig:2}(b).
In accordance with the theoretical description of Section~\ref{sec:cpa}, these points appear only in the CPA-allowed parts of the parameter space.
Let us note that the positions of the CPA can also be predicted using the proposed model by solving Eq.~\eqref{eq:cpa:sym}.

\begin{figure*}
	\centering
		\includegraphics{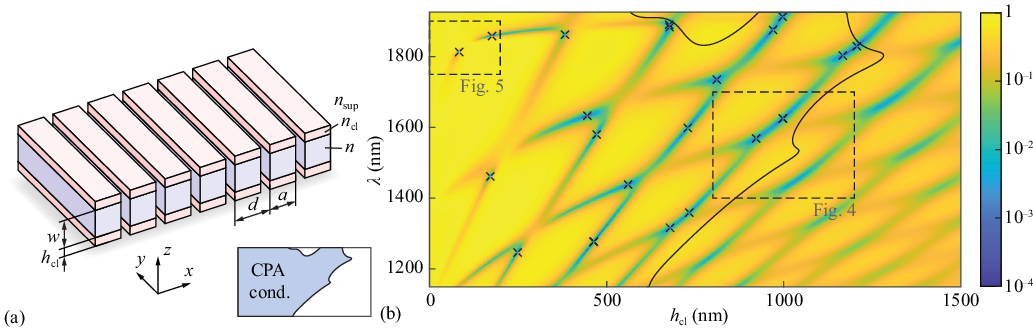}
	\caption{\label{fig:3}(a)~High-contrast grating with absorbing claddings.	
	(b)~Symmetric reflectance $|R_{\rm sym}|^2$ of an HCG with $w = 1000\nm$.
	Black crosses show the CPA points.
	The inset shows the part of the parameter space, in which the CPA condition~\eqref{eq:main:cpa} is fulfilled.
	}
\end{figure*}

\subsection{CPA in HCGs}\label{ssec:hcg}
As the second example, we will consider a high-contrast grating shown in Fig.~\ref{fig:3}(a).
The period of the considered HCG is $d=1000\nm$; its fill-factor is $a/d=1/2$.
The refractive index of air $n_{\rm sup}=1$ was used for the materials above and below the structure, as well as for the material between the grating rods.
The middle part of the HCG having the thickness~$w$ was assumed to be a dielectric with the refractive index $n = 3.5$.
The upper and lower parts of the HCG, referred to as claddings, have the thickness $h_{\rm cl}$ each and are assumed absorbing with the complex refractive index $n_{\rm cl} = n + \ii k$.
The particular values of the thicknesses $w$ and $h_{\rm cl}$, as well as of the extinction coefficient~$k$ are presented below, since they are different for different examples.
We will be studying the optical properties of the described structure in the wavelength range $\lambda\in[1150, 1925]\nm$.

A stratified medium between the claddings is a one-dimensional photonic crystal, which, for the presented parameters, supports three propagating Bloch modes.
If we, however, restrict ourselves to the case of normal incidence of TE-polarized light, which we will do, 
one of these modes will have an antisymmetric field distribution and thus will never be excited.
The two remaining modes having symmetric field distribution can be excited, which makes the considered HCG a bimodal Fabry--P{\'e}rot interferometer.
We also note that in the considered wavelength range, the HCG is subwavelength, having only one reflected and one transmitted propagating diffraction order; 
this also agrees with the assumptions made in the theoretical part of the paper.

Figure~\ref{fig:3}(b) shows the simulated intensity of the reflected light $|R_{\rm sym}|^2$ for the structure having the distance between the claddings $w=1000\nm$ and the extinction coefficient of the claddings $k=0.1$; the incident light wavelength $\lambda$ and the cladding thickness $h_{\rm cl}$ are considered as parameters.
One can see from the figure, that at several points of the parameter space, the symmetric reflection coefficient vanishes; these points marked with crosses are the points where CPA occurs.

As it is shown in the inset to Fig.~\ref{fig:3}, the CPA condition~\eqref{eq:main:cpa} is fulfilled only in a part of the parameter space.
The boundary of this part is shown with a black line in Fig.~\ref{fig:3}(b).
On this line, the inequality~\eqref{eq:main:cpa} becomes an equality,
thus the line separates the part where CPA is allowed from the part where it is forbidden.
As in the previous example, all CPA points lie exactly in the part of the parameter space where the CPA condition is fulfilled.
Let us now show what happens to the CPA points when they approach the black line.

\subsection{CPA annihilation when violating the CPA condition}

Let us consider a fragment of Fig.~\ref{fig:3}(b), the magnified version of which is shown in Fig.~\ref{fig:4}(a).
This fragment contains two CPA points, marked as CPA-1 and CPA-2.
Figures~\ref{fig:4}(b)--(d) show what happens with these CPA points when we decrease the distance between the claddings~$w$.
Since the CPA condition~\eqref{eq:main:cpa} does not contain the phases $\phi$ and $\psi$ depending on~$w$,
changing~$w$ would not change the location of the black line separating the CPA-allowed and CPA-forbidden parts of the parameter space.
It is evident from Fig.~\ref{fig:4}(a)--(c) that decreasing~$w$ makes the two CPA points to move closer to each other.
They meet at $w \approx 850\nm$ exactly on the black line [see Fig.~\ref{fig:4}(c)].
Decreasing~$w$ further makes the two CPA points disappear [see Fig.~\ref{fig:4}(d)].

The lower plots (e)--(h) in Fig.~\ref{fig:4} show the argument of the complex reflection coefficient $\arg R_{\rm sym}$.
One can see that around each CPA point, the phase exhibits vortex behavior:
the phase continuously changes from 0 to $2\pi$ (or from $2\pi$ to $0$) as we travel around a closed loop encircling a CPA.
Therefore, the phase of $R_{\rm sym}$ in undefined exactly at a CPA point.
This allows one to assign each CPA a topological invariant---topological charge---defined as~\cite{Sakotic:2021:pr}
\begin{equation}
\label{eq:tc}
C = \frac{1}{2\pi}\oint_\Gamma{\rm d}\arg R_{\rm sym},
\end{equation}
where the total differential for the considered $(h_{\rm cl}, \lambda)$ parameter space reads as
$$
{\rm d}\arg R_{\rm sym} 
= 
\frac{\partial \arg R_{\rm sym}}{\partial h_{\rm cl}} {\rm d} h_{\rm cl} 
+
\frac{\partial \arg R_{\rm sym}}{\partial \lambda} {\rm d} \lambda
$$
and $\Gamma$ is a contour in the parameter space encircling a CPA point. 
The contour integral in Eq.~\eqref{eq:tc} ``counts'' how many times the phase changes from $0$ to $2\pi$ as we travel around a CPA point.
Usually, each CPA has the topological charge $C$ equal to $+1$ or $-1$.
For example, the CPA-1 in Fig.~\ref{fig:4}(a) has the topological charge $+1$ since the phase $\arg R_{\rm sym}$ increases when the contour $\Gamma$ is traversed counterclockwise as it is shown in the figure.
Similarly, the phase decreases as we travel around CPA-2, thus its topological charge equals $-1$.
The topological charge adheres to the conservation law in the following sense: 
any interaction between the CPA points should conserve the topological charge.
In particular, the disappearance of CPA demonstrated in Fig.~\ref{fig:4} is only possible because the annihilating CPA points have opposite topological charges.

\begin{figure*}
	\centering
		\includegraphics{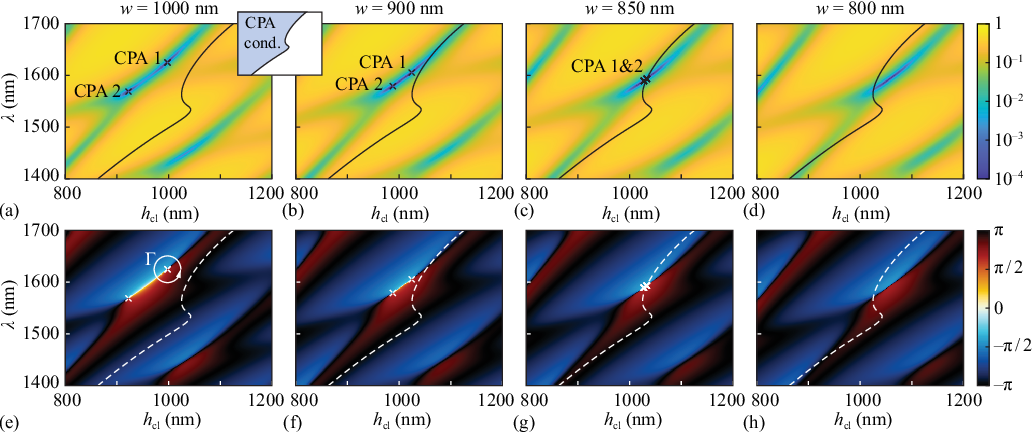}
	\caption{\label{fig:4}
	Symmetric reflection coefficient for different values of~$w$:
	intensity $|R_{\rm sym}|^2$ (upper plots)
	and phase $\arg R_{\rm sym}$ (lower plots).
	Black and white crosses show the CPA points.
	The inset shows the part of the parameter space where the CPA condition is fulfilled.
	}
\end{figure*}

\subsection{CPA at low absorption and BICs}\label{ssec:hcg:lowabs}

Let us now consider an important limiting case when the extinction coefficient~$k$ tends to zero.
In this case, the structure becomes non-absorbing, thus, no coherent perfect absorption may remain and all CPA points must annihilate, as in the previous subsection. 
To demonstrate this, we consider a different fragment of Fig.~\ref{fig:3}(b) containing another pair of CPA points.
The magnified version of this fragment is presented in Fig.~\ref{fig:5}(a).
The subplots (a)--(c) demonstrate that decreasing~$k$ from $0.1$ down to $0.01$ indeed results in CPA to move closer to each other.
If we decrease $k$ down to zero, the CPA will disappear.

One may expect that  decreasing $k$ to zero would make the black line, which separates the CPA-allowed and CPA-forbidden parts of the parameter space, to move across the parameter space, gradually ``annihilating'' all CPA points, one pair at a time.
But this is not the case; in fact, all CPA points in Fig.~\ref{fig:3}(b) group in pairs and disappear \emph{simultaneously} exactly when $k=0$.
To explain why this is the case, let us investigate the CPA condition~\eqref{eq:main:cpa} at $k=0$.

For non-absorbing structures ($k=0$), the scattering matrix $\mat{S}$ is unitary, i.\,e., its inverse is its conjugate transposed: $\mat{S}^{-1} = (\mat{S}^\TT)^*$.
Besides, the scattering matrix determinant has the modulus equal to one: $|\det\mat{S}| = 1$.
Using these facts and Eqs.~\eqref{eq:S3} and~\eqref{eq:sinv}, we can rewrite the left-hand side (LHS) of the CPA condition~\eqref{eq:main:cpa} as
$\left||r_{11}|^2 + |r_{22}|^2 - |r_0|^2 - 1\right|$.
Now let us use the fact that the norm of each row of the unitary scattering matrix $\mat{S}$ is unity: 
$$
\begin{aligned}
|r_{11}|^2+|r_{12}|^2+|t_1|^2&=1;\\
|r_{12}|^2+|r_{22}|^2+|t_2|^2&=1;\\
|t_1|^2+|t_2|^2+|r_{0}|^2&=1.\\
\end{aligned}
$$
These equalities allow us to simplify the LHS of the CPA condition~\eqref{eq:main:cpa} to $2|r_{12}|^2$, which is exactly the right-hand side (RHS) of the CPA condition.
Therefore, we have shown that for a non-absorbing structure, the inequality in the CPA condition~\eqref{eq:main:cpa} becomes an equality that is always satisfied.
This means that the CPA condition becomes violated simultaneously for the whole parameter space exactly when $k$ reaches zero.
Due to time reversal, the same holds for the lasing condition of Eq.~\eqref{eq:main:lasing}.

Now let us turn to the CPA condition~\eqref{eq:main:cpa} involving the phase matrix $\mat{E}$, i.\,e., the form of the CPA condition \emph{before} applying the theorem from the appendix.
One can easily show that condition~\eqref{eq:cpa:sym} for a unitary matrix $\mat{S}$ coincides with the lasing condition of Eq.~\eqref{eq:lasing:sym:2}, which is exactly the BIC condition~\cite{Bezus:2018:pr} once $\mat{S}$ is unitary.
Therefore, decreasing the losses in the structure to zero makes a pair of CPA points to merge and, exactly at the merging point, a BIC appears.
This is indeed the case, as we show in Fig.~\ref{fig:5}(d).
In this plot, we show not the symmetric reflectance $|R_{\rm sym}|^2$, but the conventional transmittance $|T|^2$, which is the transmitted field intensity for a single incident plane wave ($I=1$, $J=0$).
One can clearly see a resonant line with its width vanishing exactly at the point where CPAs 3 and 4 should have met.
The divergence of the $Q$-factor (not presented here) supports the fact that the considered resonant state is indeed a BIC.

\begin{figure*}
	\centering
		\includegraphics{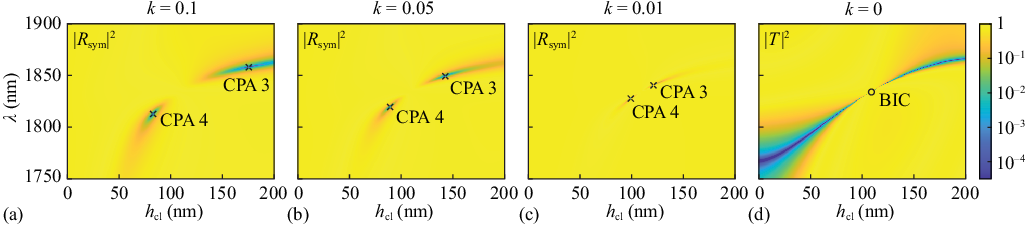}
	\caption{\label{fig:5}Symmetric reflectance $|R_{\rm sym}|^2$ at $w=1000\nm$ for different values of the extinction coefficient~$k$ (a)--(c) 
	and transmission coefficient $|T|^2$ for $k=0$.
	Black crosses show the CPA points, black circle shows the BIC.
		}
\end{figure*}

\section{Conclusion}
In the present work, we investigated the conditions of coherent perfect absorption and lasing at threshold in generalized Fabry--P{\'e}rot interferometers supporting two kinds of waves propagating between the claddings (interfaces).
We proved that CPA and lasing states in bimodal Fabry--P{\'e}rot interferometers appear only in certain regions in the parameter space, defined by inequalities directly imposed on the elements of the scattering matrices of the claddings.
Since the phase of the reflected light amplitude near CPA exhibits vortex behavior, each CPA point can be endowed with a topological charge.
We demonstrated pairwise annihilation of the CPA points with opposite topological charges occurring when the CPA condition is violated.

The CPA investigated in the present work can be referred to as conditional.
It is worth noting that another example of such ``conditional'' behavior was presented in our previous work~\cite{Bykov:2020:n}.
In that work, it was demonstrated that ``phase-only'' bound states in the continuum appear in a bimodal Gires--Tournois interferometer only when a particular inequality is satisfied.
The fact that in the present work we study not BICs but CPA in structures with a different number of open scattering channels and different symmetry,
makes the mathematical treatment in the present work completely different from Ref.~\cite{Bykov:2020:n}.
In particular, the theorem from Appendix~\ref{app:A} is not applicable to the problem considered in Ref.~\cite{Bykov:2020:n} and vice versa.
Nevertheless, the bimodal interferometer model provides a general theoretical framework allowing one to describe and study a wide variety of resonant optical effects.
In this regard, we believe that further investigation of generalized Fabry--P{\'e}rot interferometers possessing different symmetries and different number of open scattering channels is promising for explaining, why some resonant effects tend to appear in one part of the parameter space but not the other.

\vspace{1em}

\section*{Acknowledgments}
This work was funded 
by the Russian Science Foundation (project 22-12-00120-P, investigation of CPA condition and HCGs;
project 24-12-00028, investigation of ridge resonators),
and performed within the State assignment of NRC ``Kurchatov Institute'' (implementation of the numerical simulation software).

\appendix

\section{Proof of the theorem}\label{app:A}
Here, we prove two lemmas and the theorem used to obtain the lasing and CPA conditions.

\begin{lemma}\label{lemma:1}
Consider an equation for $\psi$:
$$ \alpha \ee^{2\ii\psi} + b \ee^{\ii\psi} + \alpha^*=0,$$
where $b$ is real and $\alpha$ is complex.
This equation has at least one real solution if and only if $|b| \leq 2|\alpha|$.
\end{lemma}

\begin{proof}
Let us first consider the degenerate case $\alpha = 0$. 
Here, $|b| \leq 2|\alpha|$ implies $b=0$ and any real $\psi$ is the root of the equation.
On the contrary, when $|b| > 2|\alpha|$, no real solutions exist.
This agrees with the lemma statement.

If $\alpha \neq 0$, we can consider the equation as a quadratic equation with respect to $x = \ee^{\ii\psi}$.
The roots of this equation read as
$$
x_{1,2} = \frac{-b\pm\sqrt{b^2 - 4|\alpha|^2}}{2\alpha}.
$$

When $|b| \leq 2|\alpha|$, the discriminant is non-positive and the squared modulus of $x_{1,2}$ is calculated as
$$
|x_{1,2}|^2 = \frac{b^2 + \left[-(b^2 - 4|\alpha|^2)\right]}{4|\alpha|^2}
= 1,
$$
thus $\psi$ is real as the lemma suggests.

When $|b| > 2|\alpha|$, the discriminant is positive and
$$
|x_{1,2}|^2 
= \frac{\left(b \mp \sqrt{b^2 - 4|\alpha|^2}\right)^2}{4|\alpha|^2}.
$$
Assume this equals one; therefore,
$$
b^2 - 4|\alpha|^2 =\pm b \sqrt{b^2 - 4|\alpha|^2}.
$$
This takes place either when $|b| = 2|\alpha|$ or when $\alpha = 0$, both leading to a contradiction, hence, $|x_{1,2}|^2 \neq 1$ and no real $\psi$ solves the equation.
\end{proof}

\begin{lemma}\label{lemma:2}
The following inequalities are equivalent:
\begin{equation}
\label{ineq1}
\left||a|^2+|b|^2-|c|^2-|d|^2\right| \leq 2|ab-cd|,
\end{equation}
\begin{equation}
\label{ineq2}
\left||a|^2-|b|^2+|c|^2-|d|^2\right| \leq 2|a^*c - b d^*|,
\end{equation}
where $a,b,c,d\in\mathbb{C}$.
\end{lemma}
\begin{proof}
The squared RHS of Eq.~\eqref{ineq2} can be rewritten as
$$
\begin{aligned}
4|a^*c - &b d^*|^2
=4(|ac|^2+|bd|^2-a^*c b^* d-ac^* b d^*)
\\&=4(|ac|^2+|bd|^2+|ab-cd|^2-|ab|^2-|cd|^2)
\\&=4(|a|^2-|d|^2)(|c|^2-|b|^2)+4|ab-cd|^2.
\end{aligned}
$$
The squared LHS of inequality~\eqref{ineq2} is
$$
(|a|^2-|d|^2)^2+(|c|^2-|b|^2)^2+ 2(|a|^2-|d|^2)(|c|^2-|b|^2).
$$
Therefore, inequality~\eqref{ineq2} can be written in the following equivalent form:
$$
\begin{aligned}
(|a|^2-|d|^2)^2+(|c|^2-|b|^2)^2-2(|a|^2-|d|^2&)(|c|^2-|b|^2)   \\&\leq   4|ab-cd|^2.
\end{aligned}
$$
Taking the square roots of the LHS and RHS, we obtain Eq.~\eqref{ineq1}.
Going the same way backwards allows one to obtain Eq.~\eqref{ineq2} from Eq.~\eqref{ineq1}.

\end{proof}

\begin{theorem*}\label{Theorem1}
Equation
\begin{equation}
\label{leq2}
a \ee^{\ii \phi}+b \ee^{\ii \psi}+c \ee^{\ii \phi}\ee^{\ii \psi}+d = 0
\end{equation}
with $a,b,c,d\in\mathbb{C}$
is solvable in real $\phi, \psi$ if and only if
\begin{equation}
\label{l2cond}
\left||a|^2+|b|^2-|c|^2-|d|^2\right| \leq 2|ab-cd|.\vspace{1em}
\end{equation}
\end{theorem*}
\begin{proof}
First, we prove that Eq.~\eqref{l2cond} is a \emph{necessary} condition.
Assuming $\phi$ and $\psi$ are real, we write the complex conjugate of Eq.~\eqref{leq2} multiplied by $\ee^{\ii \phi}\ee^{\ii \psi}$:
\begin{equation}
\label{leq2b}
a^* \ee^{\ii \psi}+b^* \ee^{\ii \phi}+c^* +d^* \ee^{\ii \phi}\ee^{\ii \psi} = 0.
\end{equation}

Now we take Eq.~\eqref{leq2} multiplied by $-(b^* + d^* \ee^{\ii \psi})$ and add it to Eq.~\eqref{leq2b} multiplied by $(a + c \ee^{\ii \psi})$.
This gives us the following equation:
\begin{equation}
\label{eqpsi}
(a^*c - b d^*) \ee^{2\ii \psi}
+ (|a|^2-|b|^2+|c|^2-|d|^2)\ee^{\ii \psi}
+ (a c^* - b^* d)  = 0.
\end{equation}
According to Lemma~\ref{lemma:1}, this equation has a real-$\psi$ solution only if
\begin{equation}
\label{leq3}
\left||a|^2-|b|^2+|c|^2-|d|^2\right| \leq 2|a^*c - b d^*|,
\end{equation}
which, according to Lemma~\ref{lemma:2}, is equivalent to Eq.~\eqref{l2cond}.
Therefore, we proved that the fact that Eq.~\eqref{leq2} has a real-valued solution implies Eq.~\eqref{l2cond}.

\vspace{1em}
Now let us prove that Eq.~\eqref{l2cond} is a \emph{sufficient} condition.
Assuming Eq.~\eqref{l2cond} holds, we 
use Lemma~\ref{lemma:2} to obtain Eq.~\eqref{leq3} and 
use Lemma~\ref{lemma:1} to find a real $\psi$ satisfying Eq.~\eqref{eqpsi}.
Let us show that such $\psi$ is a solution to Eq.~\eqref{leq2} for some \emph{real} $\phi$.
To do this, we express $\ee^{\ii\phi}$ from Eq.~\eqref{leq2} as
$$
\ee^{\ii\phi} = -\frac{b \ee^{\ii\psi} + d} {c \ee^{\ii\psi} + a}.
$$
We now multiply this expression by its complex conjugate and consider the quantity $1 - |\ee^{\ii\phi}|^2$, 
which, after simplification, takes the following form:
$$
1 - |\ee^{\ii\phi}|^2 = 
\frac
{N}
{|c \ee^{\ii\psi} + a|^2},
$$
where $N$ is exactly the LHS of Eq.~\eqref{eqpsi}, which is equal to zero.
This, assuming $c \ee^{\ii\psi} + a \neq 0$, proves that Eq.~\eqref{leq2} is indeed solvable in real $\phi$ and $\psi$ once the condition~\eqref{l2cond} is met, which makes it a sufficient condition.
When the denominators above vanish ($c \ee^{\ii\psi} + a = 0$), Eq.~\eqref{l2cond} is also a sufficient condition, proving which is quite trivial.

%
%
%

\end{proof}

\bibliography{CPA_Arxiv}

\end{document}